\RequirePackage{fix-cm}
\documentclass[smallextended]{svjour3}       
\smartqed  

\spnewtheorem{definitionitalic}{Definition}{\bf}{\it}
\spnewtheorem*{nonumtheorem}{Theorem}{\bf}{\it}

\usepackage{bbm}
\usepackage{bm}
\usepackage{amsmath,amssymb}
\usepackage{verbatim}
\usepackage{stmaryrd}
\usepackage{graphicx,color}
\usepackage{mathtools}
\usepackage{slashed}
\usepackage[all,cmtip]{xy}
\usepackage{tikz-cd}
\usepackage{array}
\usepackage{enumerate}

\DeclareMathOperator{\Gr}{Gr}
\DeclareMathOperator{\tr}{tr}
\DeclareMathOperator{\rank}{rank}
\DeclareMathOperator{\Der}{Der}
\DeclareMathOperator{\Cl}{Cl}
\DeclareMathOperator{\DerEnd}{MDer}
\DeclareMathOperator{\CDerEnd}{CDer}

\DeclareMathOperator{\SPAN}{span}
\DeclareMathOperator{\RE}{Re}

\DeclareMathOperator{\Aut}{Aut}

\DeclareMathOperator{\Hom}{Hom}
\DeclareMathOperator{\obj}{obj}

\DeclareMathOperator{\End}{End}
\DeclareMathOperator{\curl}{curl}
\DeclareMathOperator{\Herm}{Herm}
\newcommand{\Hermitian}{\Herm}

\DeclareMathOperator{\so}{{\mathfrak so}}
\DeclareMathOperator{\su}{{\mathfrak su}}

\newcommand{\even}{\textnormal{even}}
\newcommand{\odd}{\textnormal{odd}}

\newcommand{\tsum}{\textstyle\sum}

\newcommand{\p}{\partial}

\renewcommand{\epsilon}{USE eps INSTEAD}
\newcommand{\C}{\mathbbm{C}}
\newcommand{\R}{\mathbbm{R}}
\newcommand{\Z}{\mathbbm{Z}}

\newcommand{\W}{\mathcal{Y}}

\renewcommand{\subset}{\subseteq}
\renewcommand{\supset}{\supseteq}
\newcommand{\secref}[1]{Section \ref{sec:#1}}

\newcommand{\lemmaref}[1]{Lemma \ref{lemma:#1}}
\newcommand{\theoremref}[1]{Theorem \ref{theorem:#1}}
\newcommand{\remarkref}[1]{Remark \ref{remark:#1}}
\newcommand{\defref}[1]{Definition \ref{def:#1}}

\newcommand{\httpref}[1]{\href{http://#1}{#1}}

\newcommand{\HQ}{\mathbbm{H}}

\newcommand{\catstyle}[1]{\textnormal{\textsf{{\fontseries{sbc}\selectfont #1}}}}
\newcommand{\functorstyle}[1]{\catstyle{#1}}

\newcommand{\gx}{\mathfrak{g}}

\newcommand{\mx}{\mathfrak{m}}
\newcommand{\nx}{\mathfrak{n}}

\newcommand{\spinor}{\textnormal{spinor}}

\newcommand{\I}{\mathcal{I}}
\newcommand{\LC}{\smash{\slashed{\mathcal{L}}}}
\newcommand{\U}{\mathcal{P}}
\newcommand{\UC}{\smash{\slashed{\mathcal{P}}}}
\newcommand{\IC}{\smash{\slashed{\mathcal{I}}}}
\newcommand{\EC}{\smash{\slashed{\mathcal{E}}}}
\newcommand{\K}{\mathcal{K}}
\renewcommand{\L}{\mathcal{L}}
\newcommand{\E}{\mathcal{E}}
\newcommand{\F}{\functorstyle{E}}

\newcommand{\RicciFlat}{\functorstyle{RicciFlat}}
\newcommand{\Forget}{\functorstyle{Forget}}
\newcommand{\Spinor}{\functorstyle{Spinor}}

\newcommand{\MC}{\catstyle{MC}}

\newcommand{\lsig}{{-}{+}{+}{+}}

\newcommand{\RR}{{C^\infty}}
\newcommand{\CR}{{C^{\infty}_{\C}}}
\newcommand{\ip}[2]{\langle #1,#2\rangle}

\newcommand{\La}{\catstyle{La}}
\newcommand{\gLa}{\catstyle{gLa}}
\newcommand{\dgLa}{\catstyle{dgLa}}

\newcommand{\gLaoid}{\catstyle{gLaoid}}
\newcommand{\ggr}{\catstyle{gaugeGrpd}}
\newcommand{\ggrs}{\ggr_\spinor}
\newcommand{\dgr}{\catstyle{diffGrpd}}
\newcommand{\SET}{\catstyle{Set}}

\newcommand{\pxg}{\mathfrak{p}}

\newcommand{\oV}{\overline{V}}

\newcommand{\ww}{{\wedge W}}
\newcommand{\wwv}{{\wedge W_V}}

\newcommand{\auxa}{\lambda}
\newcommand{\auxr}{\Delta}

\newcommand{\Mf}{\catstyle{F}}
\newcommand{\Mg}{\catstyle{G}}
\newcommand{\amcb}{\mathfrak{a}}
\newcommand{\yy}{\xi}
\newcommand{\gsg}{\pxg/s\pxg}

\newcommand{\inj}{\hookrightarrow}
\newcommand{\surj}{\twoheadrightarrow}

\newcommand{\citetwo}{\cite{rt2}}

\newcommand{\diffgr}{\catstyle{diffGrpd}}
\newcommand{\oldw}{K}

\usepackage{mathptmx}      
\begin{document}

\title{The graded Lie algebra of general relativity}


\author{Michael Reiterer \and
        Eugene Trubowitz
}


\institute{Michael Reiterer,
              ETH Zurich,
              \email{michael.reiterer@protonmail.com}             \\
            \rule{55pt}{0pt} \text{Present affiliation: The Hebrew University of Jerusalem}\\
           Eugene Trubowitz,
              ETH Zurich,
              \email{eugene.trubowitz@math.ethz.ch}
}

\date{}

\maketitle

\begin{abstract}

We construct a graded Lie algebra $\mathcal{E}$ in which the Maurer-Cartan equation is equivalent to the vacuum Einstein equations.
The gauge groupoid is the groupoid of rank 4 real vector bundles with a conformal inner product, over a 4-dimensional base manifold, and the graded Lie algebra construction is a functor out of this groupoid.
As usual, each Maurer-Cartan element in $\mathcal{E}^1$ yields a differential on $\mathcal{E}$. Its first homology is linearized gravity about that element. We introduce a gauge-fixing algorithm that generates, for each gauge object $G$, a contraction to a much smaller complex whose modules are the kernels of linear, symmetric hyperbolic partial differential operators.
This contraction opens the way to  
the application of homological algebra
to the analysis of the vacuum Einstein equations.
We view general relativity, at least
at the perturbative level, as an instance of `homological PDE' at the crossroads of algebra and analysis.

\end{abstract}

\newcommand{\seed}{x_0}
\newcommand{\unk}{x}

\vskip 5mm
{\footnotesize
{\bf Important note:} This paper considerably extends and simplifies paper \cite{glagr} of the same title.
Some readers may still want to consult \cite{glagr} since it
presents some things differently or in more detail.}
\vskip 3mm

\section{Introduction} \label{sec:intro}

The moduli space of solutions to the vacuum Einstein equations is naively the set of Ricci-flat metrics of signature ${-}{+}{+}{+}$ modulo
the action of the gauge groupoid $\diffgr$ of diffeomorphisms.
For simplicity, we restrict the discussion to manifolds diffeomorphic to $\R^4$. 
In this paper, we realize  this moduli space as the set of Maurer-Cartan (MC)
elements in a graded Lie algebra (gLa) $\E$ modulo
automorphisms induced by the gauge groupoid. Informally,
\begin{equation}\label{eq:bije}
  \frac{\textnormal{Ricci-flat metrics}}{\textnormal{diffeomorphisms}}
    \;=\;
    \frac{\textnormal{nondegenerate MC-elements in $\E^1$}}{\sim}
  \end{equation}
Approximately,
an element of $\E^1$ is a pair of
a conformal orthonormal frame and a connection-like object,
in the jargon of general relativity.

The bijection \eqref{eq:bije}
sets the moduli problem for the vacuum Einstein equations,
especially at the formal perturbative level,
into the homological algebraic framework
of differential graded Lie algebras (dgLa) and $L_\infty$-algebras.
To exploit this, we provide two tools.
An algebraic tool, namely a contraction of 
the complex associated to an MC-element to a much smaller complex.
And an analytical tool,
namely a way of formulating the MC-equation
as a symmetric hyperbolic PDE.
Both tools rely on a new gauge fixing algorithm,
formulated using Clifford modules.

The power of the present formalism is in its
homological nature combined with the gauge-fixing algorithm.
There are innumerable reformulations and reinterpretations
of the vacuum Einstein equations.
We distinguish ours with concrete applications in forthcoming papers
to, for instance, the so-called Belinskii-Khalatnikov-Lifshitz (BKL)
proposal for spatially inhomogeneous singular spacetimes \cite{rt2}.

This gLa is used in \cite{nr} to study tree scattering amplitudes
for general relativity about Minkowski spacetime,
using $L_\infty$ homotopy transfer. 
\vskip 2mm

Let $\ggr$ be the category in which an object is a real, rank $4$ vector bundle
with a conformal inner product of signature ${-}{+}{+}{+}$,
over a manifold $M$ diffeomorphic to $\R^4$.
Let $W$ be the free $\RR$-module of smooth sections.
Here, $\RR$ are the real smooth functions on $M$. A morphism in $\ggr$ is a vector bundle isomorphism that preserves the conformal
inner product\footnote{%
  The category can be refined by introducing
orientations and or time orientations.}.
 The gauge groupoid $\ggr$ is deliberately bigger than the groupoid $\diffgr$ that underlies the metric formalism.
The construction of $\E$ is a functor $\F$ into the category of real, graded Lie algebras,
\[
  \F\;:\; \ggr \to \gLa
\]
More precisely, $\E$ is also a graded Lie algebroid over $\ww$.

A module derivation of $W$ is a pair of maps $\RR\rightarrow \RR$ and $W\rightarrow W$ that satisfy the Leibniz rule for both of the multiplications
$\RR\times \RR\rightarrow \RR$ and $\RR\times W\rightarrow W$. 
The module derivations
$\DerEnd_\RR(W)$ constitute a Lie algebroid over $\RR$.
Let
$\CDerEnd(W)$ be the sub Lie algebroid of module derivations
that preserve the conformal inner product.
See Definition \ref{definition:AgLaiod} for Lie algebroids.
Set  
$\mathcal{L}
=\ww\otimes \CDerEnd(W)$. All tensor products are over $\RR$.
The tensor product $\mathcal{L}$ is naturally a graded
Lie algebroid over the graded commutative algebra $\ww$ with the bracket
\[
  [\omega \delta,\omega'\delta'] = \omega\omega'[\delta,\delta'] + (\omega\auxa(\delta)(\omega'))\delta' - (\auxa(\delta')(\omega)\omega')\delta
\]
where $\auxa : \CDerEnd(W) \to \Der^0(\wedge W)$
is the canonical $\RR$-Lie algebroid morphism.
The anchor map is of type $\mathcal{L} \to \Der(\ww)$.

Let $\MC : \gLa \to \SET$ be the Maurer-Cartan functor.
\begin{nonumtheorem}[Vacuum Einstein equations as Maurer-Carten equations -- informal]
  There is a $\ww$-graded Lie algebroid ideal
  $\I \subset \L$, supported in degrees $2$ and higher,
  such that the quotient $\E = \L/\I$
  realizes the bijection \eqref{eq:bije}.
  Concretely,
  \[
    \MC(\E) = \{\unk \in \E^1 \mid [\unk,\unk]=0\}
  \]
  and $\sim$ is equivalence under automorphisms of $\MC(\E)$ 
  that are in the image of $\MC \circ \F$.
  See \secref{ricci} for Ricci-flatness.
\end{nonumtheorem}

This theorem 
places us in a standard homological algebraic framework.

The nondegeneracy referred to in \eqref{eq:bije}
is detailed in Definition  \ref{def:angh}.
Degenerate elements
are an important feature.
Since degenerate MC-elements can be easier to analyze,
it is natural to attempt to perturb degenerate
elements into nondegenerate ones, a strategy that we pursue elsewhere {\citetwo}
to study the BKL proposal.

Our construction is based on $\ww$ and this is essential.
If one replaces this by the algebra
of differential forms, the algebraic structure falls apart.
Differential forms are used,
for instance, to formulate the Yang-Mills equations.

Associated to every $\seed \in \MC(\E)$ is the
moduli space of formal perturbations
\begin{equation}\label{eq:ms2}
    \frac{\{\seed + \unk \in \MC(\E[[s]]) 
    \mid \unk \in s \E^1[[s]]\}}{\exp(s\E^0[[s]])}
  \end{equation}
with $\E[[s]]$ the gLa of formal power series in the symbol $s$
with coefficients in $\E$,
and the denominator is a group 
by the Baker-Campbell-Hausdorff formula.
One can check that the formal moduli space
\eqref{eq:ms2}
is a formal version of \eqref{eq:bije},
see also \remarkref{liefunctor}.

As usual, if $\unk \in \MC(\E)$ then the differential $d = [\unk,-]$
turns $\E$ into a dgLa.
The first homology $H^1(d)$ is  interpreted as linearized
gravity about $\unk$.
The second homology $H^2(d)$ is the obstruction space
to deformations of $\unk$
within $\MC(\E)$.
If the obstruction space vanishes, then 
the formal moduli space \eqref{eq:ms2} admits a nonlinear
parametrization by $H^1(d)[[s]]$, as reviewed in \secref{mcpt}.

Gauges are an integral part of this formalism.
By a gauge we mean a comprehensive homological object.
Here we discuss properties that all gauges will have,
a construction is in \secref{gauges1}.
Technically it is a graded $\RR$-submodule  $\E_G \subset \E$
together with certain bilinear forms,
see Definition \ref{def:gauge} for details.
An element of $\E^1$ is gauged if and only if it belongs to $\E_G^1$. 
For every future timelike $w \in W$ there is a splitting
\[
  \E = \E_G \oplus w\E_G
\]
where multiplication by
$w$ is a map of degree one, injective as a map $\E_G \to \E$.
So, $\E_G$ is a half-ranked direct summand.
Our construction of gauges
uses $\Z_2$-graded filtered Clifford modules,
and a homological unitarity trick,
based on averaging over the finite Clifford group.
We also give an algorithm that generates
such a gauge $\E_G$
for every choice of a Hermitian inner product
on a rank $18$ complex vector bundle.
 
\newcommand{\errterm}{u}
We convey the analytical content of a homological gauge with informal
statements about three systems of PDE.
These are local statements. Fix a nondegenerate $\seed \in \E^1$,
so that $\seed + \unk$ is still nondegenerate for all small
$\unk \in \E^1$. Then,
\begin{itemize}
  \item[(0)]
    For each small $\unk \in \E^1$
    the equation
    $\F(\phi)(\seed + \unk) = \seed \bmod \E_G^1$ 
    is quasilinear symmetric hyperbolic
    for an unknown automorphism $\phi \approx \mathbbm{1}$ in $\ggr$.
  \item[(1)]
    The equation
     $[\seed+\unk,\seed+\unk]=0 \bmod \E_G^2$ 
    is a quadratically nonlinear, quasilinear symmetric hyperbolic system for an unknown small $\unk \in \E_G^1$.
  \item[(2)]
    For each solution $\unk$ to (1), the equation 
    $[\seed+\unk,\errterm]=0 \bmod \E_G^3$ 
    is a linear symmetric hyperbolic system for an unknown $\errterm \in \E_G^2$.
\end{itemize}

To illustrate the utility of these systems, imagine that one attempts to
locally construct $\seed + \unk \in \MC(\E)$ with small unknown $\unk \in \E^1$.
Then,
(0) justifies restricting to $\unk \in \E^1_G$;
(1) tell us that  at least the MC-equation  modulo $\E_G^2$
is hyperbolic and can be locally solved by standard methods; and
(2) is a tool to show that the
remainder $\errterm = [\seed+\unk,\seed+\unk] \in \E_G^2$,
which solves $[\seed+\unk,\errterm]=0$
by a Jacobi identity, vanishes\footnote{%
Use the fact that
`a solution to a homogeneous hyperbolic equation that is zero initially
 is identically zero'.
 This reduces showing that $\errterm$ vanishes 
 to showing that it vanishes initially, the `constraint equations'.}.
In \cite{focus} we followed this route,
unaware of the homological framework,
to demonstrate the dynamical formation of trapped spheres
in solutions to the vacuum Einstein equations,
simplifying earlier work by Christodoulou.

The systems in (0) and (1) are nonlinear.
The analogous linear statements are simpler, and can be made globally.
\begin{nonumtheorem}[Contraction and quasiisomorphism]
Suppose $\unk \in \MC(\E)$ is globally hyperbolic, see \defref{angh}.
Define $d = [\unk,-]$ and the composition
\[
  \oldw:\quad
   \E_G \inj \E
  \smash{\;\xrightarrow{d}\;}
    \E \surj \E/\E_G
  \]
    Then $\oldw$ is a linear symmetric hyperbolic operator,
    and there is a contraction,
    hence quasiisomorphism, of complexes from
$(\E,d)$ down to $(\ker \oldw,d|_{\ker \oldw})$.
\end{nonumtheorem}
Informally, this complex
is much smaller
because it lives in three dimensions,
$\ker \oldw$ being the space of homogeneous solutions
to a linear symmetric hyperbolic system.
The contraction is a tool for calculating the homology of $d$.
At the nonlinear level,
when also applied to the bracket,
the contraction yields an $L_\infty$ algebra,
that is, 
higher many-to-one brackets on $\ker \oldw$,
most concretely using the homological perturbation lemma
also known as `homotopy transfer' \cite{b,hs}.

We have not emphasized spinors in this paper,
to avoid an extra layer of notation,
but they can be extremely useful. \secref{spinors}
is included as  a succinct discussion of the spinor functor
$\Spinor : \ggrs \to \ggr$.

We have consciously kept this paper minimalist.
Some constructions and statements generalize to other dimensions; other signatures;
topologically nontrivial base manifolds;
more general commutative rings as base rings.

\section{Related work and acknowledgments}

In general relativity,
our work is related to the Newman-Penrose orthonormal frame formalism \cite{NP}
which puts the vacuum Einstein equations in quadratically nonlinear form,
though not in Maurer-Cartan form;
and examples of gauge-fixing for the vacuum Einstein equations to
symmetric hyperbolic systems by H.~Friedrich \cite{FrH}.
The concept of symmetric hyperbolicity is
due to K.O.~Friedrichs \cite{FrKO,taylor}.
The novelty of our work is in its homological
nature; functorial constructions; and a comprehensive
concept for gauges.
In algebra, our work is  
related to deformation and obstruction theory, see
Gerstenhaber \cite{NR,Gerstenhaber}.
Applications
are bound to yield $L_\infty$-algebras and
homotopy transfer \cite{b,hs,k}.
Our own precursors include
\cite{diamond,twelve}
and parts of \cite{focus}.

We thank J.~Stasheff for emphasizing the
role of obstructions in formal perturbation theory,
and for pointing out that our contraction
can be used to run $L_\infty$ homotopy transfer.
We thank T.~Willwacher for conceptual clarifications,
and for pointing us to Lie algebroids and base change techniques,
which we have now fully adopted.


\section{Preliminaries and conventions}

It is implicitly assumed, throughout this paper, that an object of $\ggr$ is given.
The notation $M$, $\RR$, $W$ refers to such an object.
So $M \simeq \R^4$ is a manifold,
$\RR$ is the algebra of real smooth functions on it,
and $W$ is a free $\RR$-module of rank 4 with
fiberwise a conformal inner product of signature ${-}{+}{+}{+}$.

Tensor products are over $\RR$ whenever this makes sense, 
\[
    \otimes = \otimes_\RR
\]
Otherwise, tensor products are over $\R$. So $\ww$ is constructed from the tensor algebra over $\RR$.
The tensor product of elements is often denoted by juxtaposition.


\subsection{Graded Lie algebroids and representations}

The following definitions use base field $\R$ for simplicity, similar for $\C$.
A grading is a $\Z$-grading; it induces a $\Z_2$-grading.
Ungraded means concentrated in degree zero.
If $x,y$ are homogeneous then in the notation $(-1)^{xy}$
the exponent is the product of the degrees.
By $\Hom^k$ we mean morphisms that raise the degree by $k \in \Z$.

\begin{definition}[dgLa and $\MC$-functor]
  A real differential graded Lie algebra (dgLa)
  is a real graded vector space $\gx$
  with a $d \in \End^1(\gx)$
  and a $[-,-] \in \Hom^0(\gx \otimes \gx,\gx)$
  that satisfy $[x,y] = -(-1)^{xy}[y,x]$ and $d^2 = 0$ and
   $d[x,y] = [dx,y] + (-1)^x[x,dy]$
   and
   \[
     [x,[y,z]] + (-1)^{x(y+z)} [y,[z,x]] + (-1)^{z(x+y)}[z,[x,y]] = 0
   \]
  for all homogeneous $x,y,z \in \gx$.
    A real graded Lie algebra (gLa) is a dgLa with $d=0$.
    A real Lie algebra (La) is an ungraded gLa.
    Let $\La \inj \gLa \inj \dgLa$ be the corresponding categories.
    The Maurer-Cartan functor $\MC : \dgLa \to \SET$ is, on objects,
    \[
      \MC(\gx) =
        \{
            x \in \gx^1 \mid
            dx + \tfrac{1}{2}[x,x] = 0
        \}
      \]
\end{definition}
For every graded real vector space $X$,
the graded vector space $\End(X)$ is a gLa using the graded commutator.
These are the prototypical examples,
so a representation of a gLa $\gx$ is by definition a gLa morphism
$\gx \to \End(X)$ for some $X$.


A graded Lie algebroid has more structure than a gLa.
To define algebroids, let
$A$ be a unital associative graded commutative $\R$-algebra.
Graded commutative
means $ab = (-1)^{ab} ba$ for all homogeneous $a,b \in A$.
A $\delta \in \End_\R(A)$ is called a derivation
if the Leibniz rule
$\delta(ab) = \delta(a)b + (-1)^{a\delta} a \delta(b)$
holds for homogeneous elements.
The derivations $\Der(A)$ are a graded $A$-module and gLa using the graded commutator.
For a graded $A$-module, scalar multiplication must respect the grading.
We will apply this with $A = \RR$ and more interestingly $A = \ww$.
\begin{definition}[$A$-gLaoid] \label{definition:AgLaiod}
  Suppose $A$ is a unital associative graded commutative $\R$-algebra.
An $A$-graded Lie algebroid ($A$-gLaoid) is a triple $(\gx,[-,-],\rho)$
with $\gx$ a graded $A$-module; $(\gx,[-,-])$ a gLa;
the `anchor' $\rho : \gx \to \Der(A)$ is $A$-linear
and a gLa morphism;
and $[x,ay] = \rho(x)(a)y + (-1)^{ax} a[x,y]$
for all homogeneous $x,y \in \gx$ and $a \in A$.
An $A$-gLaoid morphism must be a gLa morphism, an $A$-module morphism,
and intertwine anchors.
An $A$-gLaoid ideal must be a gLa ideal, an $A$-submodule, and be
contained in the kernel of the anchor.
There is a forgetful functor $\gLaoid_A \to \gLa$.
\end{definition}
The quotient by an ideal is an $A$-gLaoid.
Note that $\Der(A)$ is an $A$-gLaoid.
\begin{lemma}[$A$-module derivations]
  With $A$ as above and $X$ a graded $A$-module,
  denote by $\DerEnd_A(X) \subset \Der(A) \oplus \End_{\R}(X)$ the elements
  $\delta = \delta_A \oplus \delta_X$ for which
  $\delta_X(ax) = \delta_A(a)x + (-1)^{a\delta} a\delta_X(x)$ for 
  all homogeneous elements.
  It is canonically an $A$-gLaoid,
  with bracket the graded commutator
  and anchor $\delta \mapsto \delta_A$.
\end{lemma}
\begin{proof}
  Omitted.
\qed\end{proof}
  When clear from context,
  we write $\delta$ for either $\delta_A$ or $\delta_X$.
  If $X$ is a faithful module, which it always is in our applications,
  then $\delta_X$ determines $\delta_A$.
\begin{definition}[$A$-gLaoid representation]
A representation of an $A$-gLaoid $\gx$ is an $A$-gLaoid morphism
$\gx \to \DerEnd_A(X)$ for some graded $A$-module $X$.
The trivial representation is the anchor map $\gx \to \Der(A)$ itself.
\end{definition}
The quotient by a $\gx$-invariant $A$-submodule
of $X$ is a new $A$-gLaoid representation.
  The adjoint representation
  $\gx \to \DerEnd_A(\gx)$ is not in general $A$-linear,
  so not an $A$-gLaoid representation.
We now define base change along a morphism $A \inj B$.
A simple example is $\R \inj \RR$,
but our main application is base change along $\RR \inj \ww$.
\begin{lemma}[Base change] \label{lemma:bcl}
  Suppose $A \inj B$ is an injective
  morphism of unital associative graded commutative $\R$-algebras,
  with $A$ ungraded.
  If $X$ is an $A$-module 
  then $X' = B \otimes_A X$ is a graded $B$-module.
  If $\gx$ is an $A$-Laoid
  and $\auxa: \gx \to \Der^0(B)$ is an $A$-Laoid morphism
  then $\gx' = B \otimes_A \gx$ is an $B$-gLaoid with bracket
  \begin{equation}\label{eq:brack}
    [by,b'y'] = bb'[y,y'] + (b\auxa(y)(b'))y' - (\auxa(y')(b)b')y
  \end{equation}
  and anchor $\gx' \to \Der(B)$, $by \mapsto (b' \mapsto b \auxa(y)(b'))$.
  In this situation, given an $A$-Laoid representation
  $\auxr : \gx \to \DerEnd_A(X)$,
  then a $B$-gLaoid representation is given by
  \[
    \gx' \to \DerEnd_B(X'),
    \quad
    by \mapsto (b'x \mapsto b\auxa(y)(b') y + bb' \auxr(y)(x))
  \]
\end{lemma}
\begin{proof}
  Omitted. Among other things,
  one has to check that the various assignments
  are well-defined, so consistent with the tensor product over $A$;
  and that the representation intertwines the (omitted) anchors.
\qed\end{proof}
Beware that while $X'$ can be viewed as just a graded $A$-module,
in general one cannot view $\gx'$ as just an $A$-gLaoid
because there is no induced anchor $\gx' \to \Der(A)$.


\subsection{Conformal structure of $W$}

As always, $W$ is the module  coming with an object of $\ggr$.

\begin{definition}[Conformally orthonormal module derivations]
  A conformal inner product is an equivalence class of inner products
  $\ip{-}{-} \in \Hom_\RR( W \otimes W, \RR)$,
  two being equivalent iff multiples by a positive function.
  Let $\CDerEnd(W) \subset \DerEnd_\RR(W)$ be the sub $\RR$-Laoid of
  all $\delta$ such that for every representative
  $\ip{-}{-}$
   there is an $f \in \RR$ with
  $\delta(\ip{x}{y})
  = \ip{\delta(x)}{y} + \ip{x}{\delta(y)}
  + f \ip{x}{y}$ for all $x,y \in W$.
\end{definition}

By a basis for $W$ we always mean a conformally orthonormal basis,
as in the next definition.
The terms basis and frame are synonymous in this context.
\begin{definition}[Conformally orthonormal basis for $W$] \label{def:uzt}
These are elements $\theta_0,\ldots,\theta_3$
that give a direct sum decomposition
$W = \RR \theta_0 \oplus \RR \theta_1 \oplus \RR \theta_2 \oplus \RR \theta_3$
such that a representative $\ip{-}{-}$
of the conformal inner product is given with $i,j=1,2,3$ by
\[
  \ip{\theta_0}{\theta_0} = -1
  \qquad
  \ip{\theta_0}{\theta_i} = 0
  \qquad
  \ip{\theta_i}{\theta_j} = \delta_{ij}
\]
Such a basis determines:
\begin{itemize}
  \item
    Elements
$\sigma_0,\sigma_1,\sigma_2,\sigma_3,\sigma_{23},\sigma_{31},\sigma_{12}
\in \End_{\RR}(W) \cap \CDerEnd(W)$ by
    $\sigma_i(\theta_0) = \theta_i$
    and $\sigma_i(\theta_j) = \delta_{ij} \theta_0$
    and $\sigma_{ij}(\theta_0) = 0$
    and $\sigma_{ij}(\theta_k) = \delta_{jk} \theta_i - \delta_{ik} \theta_j$
    and $\sigma_0(\theta_0) = \theta_0$
    and $\sigma_0(\theta_i) = \theta_i$ for all $i,j,k = 1,2,3$.
    Note that $\sigma_{ij} = -\sigma_{ji}$.
  \item 
    A map
    $\Der(\RR) \inj \CDerEnd(W)$
    producing elements that annihilate $\theta_0, \theta_1, \theta_2, \theta_3$.
    This map is given by 
    $X \mapsto (f_0 \theta_0 + \ldots + f_3 \theta_3
                \mapsto X(f_0) \theta_0 + \ldots + X(f_3)\theta_3)$.
\end{itemize}
So as $\RR$-modules we obtain
$\CDerEnd(W)
      \simeq
      \Der(\RR) \oplus \RR \sigma_0 \oplus \ldots \oplus \RR \sigma_{12}$.
  \end{definition}
    In the vernacular of relativity,
    $\sigma_0$ generates dilations;
    $\sigma_1,\sigma_2,\sigma_3$ generate boosts;
    $\sigma_{23},\sigma_{31},\sigma_{12}$ generate rotations.

\subsection{Isotypic decomposition under $\so(W)$}
Define $\so(W) = \{ \delta \in
  \End_{\RR}(W) \cap \CDerEnd(W) \mid \tr \delta = 0\}$,
    a $\RR$-Lie algebra with each fiber
    non-canonically isomorphic to $\so(1,3)$,
    and a subalgebra and ideal of $\CDerEnd(W)$.
    It consists of only vertical elements,
    meaning elements that are $\RR$-linear.
    \begin{lemma} \label{lemma:sowi}
      Suppose $X$ is a finite free $\RR$-module.
      Then every $\RR$-Laoid
      representation $\CDerEnd(W) \to \DerEnd_\RR(X)$
      has a canonical $\so(W)$-isotypic decomposition.
      Each $\so(W)$-isotypic component is invariant under
      $\CDerEnd(W)$.
    \end{lemma}
\begin{proof}
  We use the fiberwise isomorphism with $\so(1,3)$;
  the choice of the isomorphism including orientation is irrelevant.
  Define the complex isotypic projections using
  $\C \otimes \so(1,3) \simeq \su(2) \oplus \su(2)$.
  Relative to a basis for $X$,
  these are projection matrices with entries in $\C \otimes \RR$.
  By pairing them up if necessary
    we get the real isotypic projections.
    The isotypic projections
    commute with $\CDerEnd(W)$
    because $\so(W)$ is an ideal
    and its infinitesimal automorphisms, that is derivations,
    are inner.
\qed\end{proof}
Label the characters of $\su(2)$
by half-integers $p \geq 0$ with dimension $2p+1$,
hence those of $\C \otimes \so(1,3)$ by pairs of half-integers $(p,q)$.
The $\so(1,3)$-isotypic components are labeled by
$(p,p)$ respectively $(p,q) \oplus (q,p)$ with $p \neq q$, $p + q \in\Z$.

\section{The graded Lie algebra $\E = \L/\I$} \label{sec:glae}

Canonically $\DerEnd_\RR(W) \simeq \Der^0(\ww)$ as $\RR$-Laoids;
this is actually an alternative definition of module derivations of $W$.
By restriction we get a $\RR$-Laoid morphism
$\auxa: \CDerEnd(W) \to \Der^0(\ww)$ used for the base change in the next lemma.

\begin{lemma}[The gLaoid $\L$]
  Consider the $\RR$-Laoid $\CDerEnd(W)$.
  A base change along $\RR \inj \ww$, using $\auxa$ given above,
  yields the $\ww$-gLaoid
  \[
      \L = (\ww) \otimes \CDerEnd(W)
  \]
  with bracket given by \eqref{eq:brack}. As a $\ww$-module it is finite free.
\end{lemma}
\begin{proof}
  Use \lemmaref{bcl}.
\qed\end{proof}
Let $\mx$ be the kernel of the anchor map $\L \to \Der(\ww)$.
The adjoint representation $\L \to \DerEnd_{\ww}(\L)$ is not an algebroid representation,
but $\L \to \DerEnd_{\ww}(\mx)$ is.
It restricts to a $\RR$-Laoid representation
$\L^0 = \CDerEnd(W) \to \DerEnd_\RR(\mx)$.
\begin{lemma}[The ideal $\I$]
  Let $\I^2 \subset \mx^2$
be the $\so(W)$-isotypic component
\[
  (2,0) \oplus (0,2)
\]
Let $\I = (\ww)\I^2$.
Then $\I \subset \L$ is a $\ww$-gLaoid ideal.
\end{lemma}
\begin{proof}
Clearly $\I$ is a $\ww$-submodule and $\I \subset \mx$.
We have $[\L^0,\I^2] \subset \I^2$ since $\I^2$ is isotypic,
see \lemmaref{sowi},
and now $[\L,\I] \subset \I$
using the gLaoid-axioms
and $\I \subset \mx$.
This proof would have gone through for any isotypic component of $\mx^2$.
\qed\end{proof}
\begin{lemma}[The gLaoid $\E$]
  The quotient $\E = \L/\I$ is a $\ww$-gLaoid.
\end{lemma}
\begin{proof}
  Clear.
  \qed\end{proof}
\begin{remark} \label{remark:liefunctor}
If in a groupoid the automorphisms of every object
form a Lie group,
then one can associate to every object the Lie algebra of that Lie group.
The construction of $\L^0 = \E^0$ does morally just that
in the infinite-dimensional context of $\ggr$.
This 
relates the non-formal and the formal moduli spaces,
in \secref{intro}.
\end{remark}
The main goal of this section was the construction of $\E$.
In the remainder of this section we give more information about,
and alternative definitions of, the ideal $\I$.
For example, we have yet to establish that $\I \neq 0$.
\begin{lemma}[Isotypic components]
We have $\I=\I^2\oplus \I^3 \oplus \I^4$ where $\I^3 \subset \mx^3$
is the isotypic component $(\tfrac{3}{2},\tfrac{1}{2}) \oplus (\tfrac{1}{2},\tfrac{3}{2})$
and $\I^4 \subset \mx^4$ is the component $(1,0) \oplus (0,1)$.
The isotypic component
has multiplicity one for each of $\I^2$, $\I^3$, $\I^4$. So
\[
    \rank_\RR \I^2, \I^3, \I^4 = 10,16,6
\]
\end{lemma}
\begin{proof}
Use $\mx \subset (\ww) \otimes \nx$, with $\nx \subset \L^0$ the
kernel of the anchor map $\L^0 \to \Der(\RR)$.
  The unique isotypic component of $\wedge^k W$ is $(1,0) \oplus (0,1)$
  if $k=2$; $(\tfrac{1}{2},\tfrac{1}{2})$ if $k=1,3$; and
  $(0,0)$ if $k=0,4$.
  The components of $\nx$ are $(0,0)$ and $(1,0) \oplus (0,1)$.
  All have multiplicity one. The multiplication table for $\so(1,3)$ implies
  the claim.
\qed\end{proof}
\begin{lemma}[A basis for $\I^2$]
  For every choice of a conformally orthonormal basis for $W$,
  the module $\I^2 \subset \L^2$ is generated over $\RR$ by:
  \begin{equation}\label{eq:mi2}
  \RE
  \left[\begin{pmatrix}
    \theta_0\theta_1 + i \theta_2\theta_3\\
    \theta_0\theta_2 + i \theta_3\theta_1\\
    \theta_0\theta_3 + i \theta_1\theta_2
  \end{pmatrix}^T
  S
  \begin{pmatrix}
    \sigma_1 + i\sigma_{23}\\
    \sigma_2 + i\sigma_{31}\\
    \sigma_3 + i\sigma_{12}
  \end{pmatrix}
\right]
  \end{equation}
  where $S \in \C^{3 \times 3}$ runs over all symmetric traceless matrices.
  Here $\theta_0\theta_1 = \theta_0 \wedge \theta_1$.
\end{lemma}
\begin{proof}
  Direct computation.
  These elements annihilate $\RR$,
  and by the properties of $S$ they annihilate $W$,
  hence they are in $\mx^2$.
  Roughly,
  the two given vectors are separately $(1,0)$-representations,
  and the construction picks out the desired isotypic component
  for $\I^2$.
  This is related to what is called, in the vernacular of relativity,
  the electromagnetic decomposition of the Weyl tensor.
\qed\end{proof}
We give further independent constructions of $\I$:
\begin{itemize}
  \item Using a deformation argument, in \secref{cliffmod}.
  \item Using spinors, in \secref{spinors}.
\end{itemize}

\section{Theorems relating to Ricci-flatness} \label{sec:ricci}

An object of $\ggr$ is fixed,
in particular  we have a manifold $M \simeq \R^4$
and we have constructed a $\ww$-gLaoid $\E$.
Let $\Omega = \wedge \Omega^1$ be the
graded commutative $\RR$-algebra of differential
forms on $M$. Base change yields an $\Omega$-gLaoid
\[
    \W = \Omega \otimes \DerEnd_\RR(\Omega^1)
\]
\begin{definition}[Affine, nondegenerate, globally hyperbolic] \label{def:angh}
  We use the fact that
  $\Omega^1$ and $\Der(\RR)$
  are canonically dual  as $\RR$-modules.
  We say:
  \begin{itemize}
\item $\nabla \in \W^1$ is
   affine iff $\nabla|_\RR \in \Omega^1 \otimes \Der(\RR)$
   is the identity $\Omega^1 \to \Omega^1$.\\
 Affine $\nabla$ are in canonical one-to-one correspondence with affine connections.
    \item  $\unk \in \E^1$ is nondegenerate
iff $\unk|_\RR \in W \otimes \Der(\RR)$ is an
isomorphism $\Omega^1 \to W$.\\
Such an isomorphism can be interpreted as a frame for the tangent bundle. 
    \item $\unk \in \E^1$
      is globally hyperbolic iff there exists a diffeomorphism
      $M \to \R^4$ such that,
      with $t,\xi_1,\xi_2,\xi_3 \in \RR$ the four coordinate functions,
      we have $\unk(t + \sum_i n^i \xi_i) \in W_+$
      for all constants $(n^1,n^2,n^3) \in \R^3$ with
    $(n^1)^2 + (n^2)^2 + (n^3)^2 \leq 1$.
  \end{itemize}
\end{definition}

We do not need globally hyperbolic in this section.
Our definition of global hyperbolicity is
stronger than necessary, for the sake of simplicity.
The way we have defined it,
globally hyperbolic does not imply nondegenerate.

\begin{theorem}[From an MC-element to a Ricci-flat metric] \label{theorem:mctoricci}
  For a nondegenerate $\unk \in \E^1$
  denote by $i: W \to \Omega^1$
  the inverse of $\unk|_\RR$.
  It induces a gLa map $\L \to \W$.
  Define $\nabla$ to be the image of $x$ under
  \[
    \E^1 = \L^1 \to \W^1
  \]
  Then $\nabla$ is affine,
  and compatible with the $i$-induced conformal metric on $M$,
  so every representative metric $g \in S^2\Omega^1$ satisfies $\nabla g \in \Omega^1 g$.
  If in addition $\unk \in \MC(\E)$, then $\nabla$ is torsion-free;
  there is a unique representative $g$,
  up to a positive multiplicative constant,
  that satisfies $\nabla g = 0$;
  and this metric $g$ of signature $\lsig$ is Ricci-flat.
\end{theorem}
\begin{proof}
  Let $h \in S^2W$ be the inverse of a representative
  of the conformal inner product on $W$.
  Then $xh \in Wh$ in $W \otimes S^2W$,
  by the definition of $\CDerEnd(W)$.
  The $i$-induced conformal metric $g$ satisfies $\nabla g \in \Omega^1 g$.
  The choice of $i$ makes $\nabla$ affine,
  so it corresponds to an affine connection
  and one can speak about torsion.
  Then $\nabla$ is torsion-free iff $[\nabla,\nabla]|_\RR = 0$.
  And if $\nabla$ is torsion-free then
  $[\nabla,\nabla]: E \to \Omega^2\otimes E$ is $\RR$-linear and it is
  the curvature of the connection $\nabla: E \to \Omega^1 \otimes E$
  induced on (the only two cases we need)
  $E = \RR g$
  and $E = \Omega^1$.
  If $\unk \in \MC(\E)$ then $[\nabla,\nabla] \in \K$ with
   $\K$ the image of $\I^2$ under $\L^2 \to \W^2$.
  The definition of $\I^2$ and hence $\K$ implies that
  $\nabla$ is torsion-free;
  that $\nabla$ has vanishing curvature as a connection on
  the rank one module $E = C^\infty g$
  which yields existence and uniqueness of a
  new representative $g$ as stated
  with $\nabla$ its Levi-Civita connection;
  that this new $g$ is Ricci-flat using $E = \Omega^1$.
\qed\end{proof}

Let $\dgr$ be the category of manifolds diffeomorphic to $\R^4$
and diffeomorphisms.
There is a forgetful functor
$\Forget: \ggr \to \dgr$.
By assigning to every manifold
the set of Ricci-flat metrics of signature $\lsig$
over it, we get a functor
$\RicciFlat : \dgr \to \SET$.
Let
$\F: \ggr \to \gLa$ be the construction of $\E$.

\begin{theorem}[Equivalent moduli spaces]
  For all $X \in \obj(\ggr)$:
  \[
      \frac{\RicciFlat(\Forget(X))}{\sim} \;\simeq\;
      \frac{
        \textnormal{nondegenerate elements in } \MC(\F(X))}{\sim}
  \]
  where quotient by $\sim$ means modulo automorphisms
  in the image of $\RicciFlat$ respectively in the image of $\MC \circ \F$.
  The bijection $\simeq$ is induced by \theoremref{mctoricci}.
\end{theorem}
\begin{proof}
  Omitted.
\qed\end{proof}


\section{Clifford modules as deformations}\label{sec:cliffmod}

Some constructions in this paper,
perhaps more than we are aware,
are naturally stated using Clifford algebras and modules \cite{abs}.
Clifford algebras are $\Z$-filtered and $\Z_2$-graded.
Accordingly, Clifford modules can be filtered or $\Z_2$-graded.
In \secref{cmod} we review the highly constrained structure
of such Clifford modules.
Beware that by a graded or filtered module
we mean one where module multiplication respects the
grading or filtration respectively.

This section has two goals.
One is to show that $\E$ is free as an unfiltered Clifford module,
which we need for gauges in \secref{gauges1}.
The other is a new and clean definition of the ideal $\I \subset \L$
using a deformation argument.

Suppose $R$ is a filtered ring and $\Gr R$ is its associated graded ring.
Then $\Gr$ is also a functor from filtered $R$-modules to graded $\Gr R$-modules.
We will use this intuitive fact:
If $f: A \to B$ is a morphism of filtered $R$-modules,
then by `semicontinuity' the kernel of $\Gr f : \Gr A \to \Gr B$
should not be smaller than the kernel of $f$.
A rigorous version is that if $k: K \to A$ satisfies $f \circ k = 0$
and if $k$ is a split monomorphism,
then $\Gr f \circ \Gr k = 0$
and $\Gr k: \Gr K \to \Gr A$ is a split monomorphism.
Here $k$ being a split monomorphism means 
equivalently that
it has a left-inverse $A \to K$ that is also a map of filtered modules,
equivalently that
$K$ is (via $k$) a direct summand of $A$ as a filtered module.
So being a split monomorphism is stronger than being injective.

To define the Clifford algebra $\Cl(W)$ we
need a representative $\ip{-}{-}$
of the conformal inner product on $W$,
though the dependence on the representative is minor.
A basis $\theta_0,\ldots,\theta_3 \in W$ is always understood to be orthonormal
for this representative.

\begin{definition}[Clifford algebra]
  For a chosen representative $\ip{-}{-}$,
  let $\Cl(W)$ be the free associative $\RR$-algebra generated by $W$
  modulo the two-sided ideal generated by
    $wv+vw + 2\ip{v}{w}$
    for all $v,w \in W$.
    It has a canonical filtration $\Cl(W)^{\leq k}$ and
  compatible $\Z_2$-grading. We abbreviate $\Cl = \Cl(W)$.
\end{definition}
The map $W \inj \Cl$ is injective,
and $\Gr\,\Cl = \ww$ as $\RR$-algebras.
Hence $\Gr: \Mf \to \Mg$
with $\Mf$ the category of filtered
and compatibly $\Z_2$-graded $\Cl$-modules;
$\Mg$ the category of graded $\ww$-modules.
If $X \in \obj(\Mf)$ then we
set $X^{\lhd k} = X^{\leq k} \cap X^\odd$ for $k$ odd,
    $X^{\lhd k} = X^{\leq k} \cap X^\even$ for $k$ even.
  Filtration and $\Z_2$-grading being compatible means
  $X^{\leq k} = X^{\lhd k} \oplus X^{\lhd k-1}$,
  in particular $\Gr X = \oplus_k X^{\leq k}/X^{\leq k-1}
  \simeq \oplus_k X^{\lhd k}/X^{\lhd k-2}$.
  \begin{lemma}[The Clifford module $\UC$] \label{lemma:ucl}
  Abbreviate $\Omega = \wedge^4 W$. Define
  \[
    \UC^{\lhd k} \,\subset\, \Hom_\R(\RR,\Cl^{\lhd k})
    \oplus \Hom_\R(W,\Cl^{\lhd k+1}) \oplus \Hom_\R(\Omega,
    \Cl^{\lhd k} \otimes \Omega)
\] 
to be the elements $\delta_\RR \oplus \delta_W \oplus \delta_\Omega$ for which\footnote{%
Juxtaposition is multiplication in $\RR$ or $\Cl$,
or scalar multiplication for a $\RR$-module,
and the injection $W \inj \Cl$ is implicit.}
\begin{align*}
  \delta_\RR(ff')
& = f'\delta_\RR(f) + f \delta_\RR(f')\\
\delta_W(fw) & = \delta_\RR(f)w + f\delta_W(w)\\
\delta_\Omega(f\eta) & = \delta_\RR(f) \otimes \eta + f \delta_\Omega(\eta)
\end{align*}
Then $\UC$ is a $\Cl$-module, and an object in $\Mf$.
It is free of rank 9 if the filtration is ignored,
$\UC \simeq \Cl^9$ as unfiltered $\Z_2$-graded $\Cl$-modules.
The $\ww$-module $\U = \Gr \UC$
is given in the same way, by
syntactically replacing $\Cl$ by $\ww$, and obvious grading.
We have
$\rank \UC^{\lhd 0,1,2,3,4} = 21, 48, 67, 72, 72$.
\end{lemma}
\begin{proof}
    Let $\Cl'$ be the space $\Cl$ with opposite $\Z_2$-grading,
    without filtration,
  then we have $\Cl' \simeq \Cl$ as unfiltered $\Z_2$-graded $\Cl$-modules;
  this statement fails for $\ww$.
  One shows that
  $\UC \simeq (\Cl \otimes \Der(\RR)) \oplus
  \Hom_\RR(W,\Cl')
  \oplus
  \Hom_\RR(\Omega,\Cl \otimes \Omega)$,
using a basis for $W$ and the Leibniz rules defining $\UC$.
\qed\end{proof}
\begin{lemma}[The morphism $f$]
 Set $\LC = \Cl \otimes \CDerEnd(W) \in \obj(\Mf)$.
 In $\Mg$ we have a canonical $\Gr \LC \simeq \L = \ww \otimes \CDerEnd(W)$.
 There is a morphism in $\Mf$ given by
 \[
        f : \LC \to \UC
 \qquad
 \omega \delta \mapsto \delta_\RR \oplus \delta_W \oplus \delta_\Omega
 \]
 where $\delta_\RR(f) = \omega \delta(f)$
 and $\delta_W(w) = \omega \delta(w)$
 and $\delta_\Omega(\eta) = \omega \otimes \delta(\eta)$.
 Then:
 \begin{itemize}
   \item The morphism $f$ is surjective,
     and so is $\LC^{\lhd k} \to \UC^{\lhd k}$ for $k=2,3,4$.
   \item Set $\IC = \ker f$ and $\EC = \LC/\IC$.
     Then $\EC \simeq \UC$ as unfiltered $\Z_2$-graded $\Cl$-modules. 
   \item $\Gr f : \L \to \U$
     is given by syntactically replacing $\Cl$ by $\ww$.
     It is not surjective.
  \end{itemize}
\end{lemma}
\begin{proof}
  The surjectivity claim reduces to checking surjectivity for
  $k=2,3,4$ for the $\RR$-linear map
  $\Cl^{\lhd k} \otimes \so(W) \to \Hom_{\RR}(W,\Cl^{\lhd k+1})$,
  $\omega \delta \mapsto (w \mapsto \omega \delta(w))$
  which is by direct calculation;
  the ranks
  are $6,24,42,48,48$ on the left and $16,28,32,32,32$
  on the right for respectively $k=0,1,2,3,4$.
  Beware that $\EC \simeq \UC$ is not in $\Mf$,
    in fact $\UC \to \EC$ is not filtered
    since say $f: \LC^{\lhd 0} \to \UC^{\lhd 0}$ is not surjective.
  Note that $\Gr f$ is as claimed only because the filtration
  of $\UC$ is set up correctly.
\qed\end{proof}
\begin{lemma}[Properties of $\IC$] We have:
  \begin{itemize}
\item $\IC$ is a direct summand of $\LC$
  in $\Mf$, 
  free unfiltered $\Cl$-module
  of $\RR$-rank 32.
    \item $\IC^{\leq 1} = 0$
  and the elements \eqref{eq:mi2},
  now interpreted in $\LC^{\lhd 2}$, are a $\RR$-basis of $\IC^{\lhd 2}$.
\item $\IC = \Cl \IC^{\lhd 2}$,
  and $\rank \IC^{\lhd 0,1,2,3,4} = 0,0,10,16,16$.
\item 
  We have a split
  short exact sequence
  $0 \to \IC \to \LC \to \EC \to 0$ 
  in $\Mf$.
  \end{itemize}
\end{lemma}
\begin{proof}
  Rank $32$ since $f$ is surjective.
  By $\rank \LC^{\lhd 0,1,2,3,4}
  = 11,44,77,88,88$
  we get
  $\rank \IC^{\lhd 2,3,4}
  = 10,16,16$.
  It suffices to check $\rank \IC^{\lhd 0,1} = 0,0$
  which we omit.
  In the third claim,
  both sides are free $\Cl$-modules by
  \theoremref{pf}, so their ranks are multiples of 16,
  evenly distributed on even and odd parts.
  Inclusion $\supset$ is clear,
  and $\subset$ follows from $\rank \IC^{\lhd 2} = 10 > 8$.
\qed\end{proof} 
\begin{theorem}[Associated gradeds and new definition of $\I$] \label{theorem:nd}
  In $\Mg$
  the associated graded $\Gr \IC$ is a direct summand of $\Gr \LC$.
  In $\Mg$ we have a canonical isomorphism
  $i : \Gr \LC \to \L = (\ww) \otimes \CDerEnd(W)$.
  Define afresh, $\I = i(\Gr \IC)$. Then
  \begin{itemize}
    \item $\I^0 = \I^1 = 0$
  and $\I^2$ has $\RR$-basis \eqref{eq:mi2},
  and $\I = (\ww) \I^2$.
\item $\I$ is contained in the kernel of
  the anchor map $\L \to \Der(\ww)$ of $\L$.
\item $[\L,\I] \subset \I$.
  \end{itemize}
  Define afresh, $\E = \L/\I \simeq \Gr \LC/\Gr \IC \simeq \Gr \EC$.
  It is a $\ww$-gLaoid.
\end{theorem}
\begin{proof}
  The isomorphism $i$ is induced from
  the isomorphism $\Gr \Cl \to \ww$.
  By construction and semicontinuity,
  the newly defined $\I$ is contained in the kernel of
  $\Gr f$ hence in the kernel of the 
    anchor map $\L \to \Der(\ww)$.
    It is not difficult to see that
    $[\L^0,\I^2] \subset \I^2$.
    Together it follows that $[\L,\I] \subset \I$.
    The rest is omitted.
\qed\end{proof}
The definition of $\I$ in \theoremref{nd} matches
the old one, in \secref{glae}. The point of the new definition
is that one can prove all the main properties independently.

\begin{theorem}[Freeness as unfiltered Clifford modules] \label{theorem:fr}
There are $\RR$-submodules $A \subset \LC^{\lhd 0}$ and
$B \subset \IC^{\lhd 2} \subset \LC^{\lhd 2}$,
free
of ranks $9$ and $2$ respectively, 
such that 
\[
\LC \simeq \Cl \otimes (A \oplus B)
\qquad
\IC \simeq \Cl \otimes B
\qquad
\EC \simeq \Cl \otimes A
\]
as $\Z_2$-graded $\Cl$-modules (not necessarily in $\Mf$)
with isomorphism
$\omega x \mapsfrom \omega \otimes x$.
\end{theorem}
\begin{proof}
  Use the morphism $f$ and \lemmaref{ucl}, or use 
  \theoremref{pf},
  Explicitly, one can take
  $A = \Der(\RR) \oplus
  \SPAN_{\RR} \{\sigma_0,\sigma_1,\sigma_2,\sigma_3,\sigma_{23}\}$.
\qed\end{proof}
The freeness of $\EC$ is exploited in \secref{gauges1}.
Beware that $\E$ is not free over $\ww$, indeed a free $X$
must necessarily satisfy $(\wedge^4 W)X \simeq X/WX$,
whereas $\E^4 \not\simeq \E^0$.

As algebras $\Gr \Cl = \ww$,
but we have so far consciously suppressed the well-known fact that as $\RR$-modules
there is even a canonical $\Cl \simeq \ww$,
\lemmaref{cis}.
Hence $\Cl$ acquires a module $\Z$-grading,
and if $\Cl_1$ and $\Cl_2$ are defined
using two representatives of the conformal inner product,
then there is still a canonical $\RR$-module isomorphism
$\Cl_1 \simeq \ww \simeq \Cl_2$.
\begin{lemma} \label{lemma:cis}
A $\Cl$-module structure on $\ww$
is induced by $W \to \End(\ww)$, $w \mapsto e_w + i_w = c_w$ where 
$e_w \in \End^1(\ww)$ is multiplication by $w \in W$ and
$i_w \in \End^{-1}(\wedge W)$ is defined by
$i_we_v + e_vi_w = -\ip{v}{w}$ for all $v \in W$.
As $\Cl$-modules, $\Cl \simeq \ww$.
\end{lemma}
\begin{proof}
We have $e_we_v + e_ve_w = i_wi_v + i_vi_w = 0$,
hence $c_wc_v + c_vc_w + 2\ip{v}{w} = 0$.
\qed\end{proof}
In an orthonormal basis, the identification is
$\theta_{i_1} \cdots \theta_{i_k} 
\mapsto \theta_{i_1} \wedge \cdots \wedge \theta_{i_k}$
for $i_1 < \ldots < i_k$.


\section{Gauges, definition and construction} \label{sec:gauges1}

We start with a purely algebraic definition of a gauge, \defref{gauge}.
These are comprehensive gauges in all degrees of $\E$,
suitable for homology, and 
designed for compatibility with the PDE concept of symmetric hyperbolicity,
see \secref{shs}.

To show that gauges as in \defref{gauge} actually exist,
we use the Clifford module $\EC$ from \secref{cliffmod}.
This entire section only depends on the fact
that $\EC$ is a filtered $\Z_2$-graded $\Cl$-module,
meaning $\EC$ is in $\Mf$,
free as an unfiltered $\Z_2$-graded $\Cl$-module,
and $\E = \Gr \EC$ as graded $\ww$-modules, so in $\Mg$.
As before, $\Cl = \Cl(W)$.
When using Clifford modules we implicitly use
a representative $\ip{-}{-}$ of the conformal inner product,
but the dependence on it is completely minor;
we will not dwell on this.
The account given here is a consolidated one,
based on \cite{twelve,glagr}.

Let $W_+ \subset W$ be the nonempty set of all
elements that are everywhere future timelike,
this requires the choice of a time direction.
For example, using a conformally orthonormal basis,
$W_+ = \{ \sum_i w_i\theta_i \mid w_0 > (w_1^2 + w_2^2 + w_3^2)^{1/2}
          \; \text{in $\RR$} \}$.
Set $\Hom = \Hom_\RR$; continue to set $\otimes = \otimes_\RR$;
and let $S^2$ be the symmetric tensor product over $\RR$.
\begin{definition}[Gauge] \label{def:gauge}
  A gauge is a pair $(\E_G,B)$.
  A graded finite free $\RR$-submodule $\E_G \subset \E$
  such that for every $w \in W_+$, left-multiplication
  $w: \E_G \to \E$ is injective and
  \begin{equation}\label{eq:zzu1}
    \E = \E_G \oplus w\E_G
  \end{equation}
  so necessarily $\E_G$ must have half the rank of $\E$,
  and $\E_G^0 = \E^0$ and $\E_G^4 = 0$.
  And for every $k$ an element $B^k \in \Hom(\E_G^k \otimes \E^{k+1}, \RR)$
  with:
  \begin{itemize}
    \item[(a)] $B^k(-,w-)|_{\E_G^k \times \E_G^k} \in \Hom(S^2 \E_G^k,\RR)$ for all $w \in W$, a symmetry requirement.
    \item[(b)] This is positive definite whenever $w \in W_+$.
    \item[(c)] $\E_G^{k+1} = \{ x \in \E^{k+1} \mid B^k(\E_G^k,x) = 0\}$.
  \end{itemize}
\end{definition}

We take $b \in \Hom^\odd(S^2 \EC,\RR)$ to mean
$b(\EC^\odd,\EC^\odd)=b(\EC^\even,\EC^\even)=0$.
Multiplication by $w$ 
uses the $\ww$-module structure in (a), the $\Cl$-module structure in (i).
\begin{theorem}[Sufficient linear problem] \label{theorem:sp}
  Suppose $b \in \Hom^\odd(S^2\EC,\RR)$ satisfies:
  \begin{itemize}
    \item[(i)] $b(-,w-) \in \Hom^{\even}(S^2\EC,\RR)$
          for all $w \in W$, a symmetry requirement.
      \item[(ii)] This is positive definite whenever $w \in W_+$.
  \end{itemize}
  Define
  \[
    \EC_G^k \;=\; \{ x \in \EC^{\lhd k} \mid b(x,\EC^{\lhd k-1})=0 \}
  \]
  Then for every
  $w \in W_+$,
  Clifford left-multiplication
  $w: \EC^{\lhd k-1} \to \EC^{\lhd k}$ is injective and
  \begin{equation}\label{eq:zzu}
          \EC^{\lhd k} = \EC_G^k \oplus w \EC^{\lhd k-1}
        \end{equation}
  The map $b$ induces a map
  \[
    b^k \in \Hom(\EC_G^k \otimes (\EC^{\lhd k+1} / \EC^{\lhd k-1}),\RR)
  \]
  We have
  $\EC_G^k \cap \EC^{\lhd k-2} = 0$.
  Let $\E_G^k$ be the isomorphic image of $\EC_G^k$
  under the canonical
  surjection $p^k: \EC^{\lhd k} \to \EC^{\lhd k}/\EC^{\lhd k-2} \simeq \E^k$.
  Let $B^k \in \Hom(\E_G^k \otimes \E^{k+1},\RR)$
      be the map corresponding to $b^k$.
      Then this defines a gauge as in \defref{gauge}.
\end{theorem}
\begin{proof} 
  Clifford left-multiplication by $w \in W_+$ is injective
  since $w^2$ is a nonzero multiple of the identity.
  To prove \eqref{eq:zzu}
  show that the intersection of the summands vanishes using (ii)
  and make a rank argument again using (ii).
    Fix a $w \in W_+$.
  If $x \in \EC_G^k \cap \EC^{\lhd k-2}$
  then $b(x,w x) = 0$,
  so by (ii) we get $x = 0$.
  Applying \eqref{eq:zzu} twice gives \eqref{eq:zzu1}, because
  \[
    \EC^{\lhd k}
    = \EC^k_G \oplus w \EC^{k-1}_G \oplus w^2 \EC^{\lhd k-2}
  \]
 and $w^2$ is a nonzero multiple of the identity.
  Note that $p^{k+1} w = w p^k$
  as maps $\EC^{\lhd k} \to \E^{k+1}$.
  For every $x \in \EC^k$ let
  $x' = p^k x$,
  so $x \mapsto x'$ is bijective as a map
  $\EC_G^k \to \E_G^k$.
  For $x \in \EC^k_G$ and $y \in \EC^k$ we have
  $B^k(x',wy') = b^k(x,wy')
      = b^k(x,(wy)') = b(x,wy)$,
      then restrict to $y \in \EC^k_G$,
      to get (a) and (b), and (c) by a rank argument.
\qed\end{proof}
\begin{remark}\label{remark:init}
  Condition (i) would be easy to satisfy if it was only required
  for a single $w$, say for $w = \theta_0$.
  In fact, there is a bijection between:
  \begin{itemize}
    \item The set of $b' \in \Hom^\odd(S^2\EC,\RR)$
  for which $b'(-,\theta_0-) \in \Hom^\even(S^2\EC,\RR)$.
\item The set of $b'' \in \Hom(S^2\EC^\even,\RR)$.
  \end{itemize}
  Furthermore $b'(-,\theta_0-)$ is positive definite
  if and only if $b''$ is positive definite.
  The map $b' \mapsto b''$ is given by $b''(x,y) = b'(x,\theta_0y)$ for even $x,y$.
  The inverse $b'' \mapsto b'$ is given by
  $b'(x,y) = b'(y,x) = b''(x,\theta_0y)$ for even $x$, odd $y$.
  Use $(\theta_0)^2 = 1$ in $\Cl$.
\end{remark}
The following theorem can be used to construct $b$ that satisfy
(i) and that partially satisfy (ii), in a way that is still useful.
\begin{theorem}[%
  Invariant Clifford average, Clifford unitarity trick]
  The invariant Clifford averaging element
  $\pi \in S^2 \Cl$ in \theoremref{ica}
  defines a $\Pi \in \End^\even(S^2 \EC)$.
  Suppose $b'$
  is as in \remarkref{init} with $b'(-,\theta_0-)$ positive definite.
  Then
  \[
      b = b' \circ \Pi
  \]
  satisfies (i) and $b(-,\theta_0-)$ is positive definite.
  And this is a projection, in the sense
   that $b=b'$ if and only if $b'$ already satisfied (i).
\end{theorem}
\begin{proof}
  We have $b \in \Hom^\odd(S^2\EC,\RR)$ since
  $\pi$ is even. Use \theoremref{ica}.
    Positivity since
    $b(-,\theta_0-) = \frac{1}{|F|} \sum_{f \in F} b'(f-,\theta_0f-)$
  is an average without signs.
\qed\end{proof}
We now
parametrize more explicitly
the space of $b$ that satisfy the assumptions of \theoremref{sp}.
These assumptions are oblivious to the filtration of $\EC$,
only its structure as a
$\Z_2$-graded $\Cl$-module counts,
so we can use the isomorphism in \theoremref{fr}.
The rank of $A$ plays a minor role in the following.

The `transpose' $x \mapsto x^T$ is the unique anti-automorphism of $\Cl$
that acts as the identity on the image of $W \inj \Cl$.
As a $\RR$-module,
$\Cl$ has a canonical $\Z$-grading by \lemmaref{cis},
the degree $k$ subspace
having basis $\{ \theta_{i_1}\cdots \theta_{i_k} \mid i_1 < \ldots < i_k\}$.
Let
\[
  \langle-\rangle_\# : \Cl \to (\C^2 \otimes_\C \C^2) \otimes \RR
\]
be the unique $\RR$-linear map that annihilates elements of even degree,
and $\langle \theta_i \rangle_\# = \sigma_i$
for $i=0\ldots 3$
and $\langle \theta_1\theta_2\theta_3 \rangle_\# = -i\sigma_0$,
$\langle \theta_0\theta_2\theta_3 \rangle_\# = -i\sigma_1$,
$\langle \theta_0\theta_3\theta_1 \rangle_\# = -i\sigma_2$ and
$\langle \theta_0\theta_1\theta_2 \rangle_\# = -i\sigma_3$
where $\sigma_i \in \Herm(\C^2) \subset \C^2 \otimes_\C \C^2$ are the Pauli matrices.
Below, $\Herm(\C^2 \otimes A)$ are the $\RR$-bilinear
Hermitian forms, antilinear in the first argument.

\begin{theorem}[Explicit construction of gauges] \label{theorem:m2h}
     We use $\EC \simeq \Cl \otimes A$ from \theoremref{fr}.
     An isomorphism of $\RR$-modules
  \[
      \Herm(\C^2 \otimes A) \to
      \left\{
        \begin{array}{l}
          b \in \Hom^\odd(S^2 \EC,\RR)\;\text{with}\\
          b(-,w-) \in \Hom^\even(S^2 \EC ,\RR)
\;\text{for all $w \in W$}
\end{array}\right\}
\]
is given by
$h \mapsto b_h$ where for all $x,x' \in \Cl$ and $a,a' \in A$:
\[
  b_h(xa,x'a') = \RE\big(h(-\otimes a,-\otimes a') (\langle x^T x' \rangle_\#)\big)
\]
If $h$ is positive definite,
then $b_h(-,w-)$ is positive definite for all $w \in W_+$.
\end{theorem}
\begin{proof}
  Note that $\langle x^T \rangle_\#$ is the conjugate transpose
  of $\langle x \rangle_\#$ for all $x \in \Cl$.
  Therefore $(x^Tx')^T = (x')^Tx$ and $(x^T w x')^T = (x')^T w^T x = (x')^T w x$
  imply the symmetry of $b_h$ and $b_h(-,w-)$ respectively.
  By a linear algebra computer calculation the map is an isomorphism,
  in particular the space of $b$
  and the space $\Herm(\C^2 \otimes A)$ have equal rank $324 = 18^2$.
  We sketch how positivity is proved.
  By $\textnormal{SL}(\C^2)$-symmetry, it suffices to check positivity for $w = \theta_0$.
  It suffices to check that 
  $f: \Cl_{13} \times \Cl_{13} \to \C^2 \otimes_\C \C^2$,
  $(x,x') \mapsto \langle x^T \theta_0 x' \rangle_\#$
  is of the form
  $(x,x') \mapsto \sum_B \overline{B x} \otimes_\C B x'$
  for a finite set of
  $B \in \Hom_\R(\Cl_{13},\C^2)$ whose
  common kernel vanishes.
  Since $f$ annihilates (odd,even),
  consider (even,even) only, (odd,odd) is similar.
  Parametrize $u: \R^4 \oplus \R^4 \to \Cl_{13}^\even$,
  $v\oplus w \mapsto
  w_0 + v_1 \theta_0\theta_1
  + v_2 \theta_0 \theta_2 + v_3 \theta_0\theta_3
  + w_1 \theta_2\theta_3 + w_2\theta_3\theta_1 + w_3\theta_1\theta_2
  - v_0 \theta_0\theta_1\theta_2\theta_3$.
  A calculation shows that $f(u(v \oplus w),u(v' \oplus w'))$
  equals $\sum_e\; \overline{B_e(v-iw)} \otimes_\C B_e(v'-iw')$
  times a positive constant,
  with summation over the 16 elements
  $e= (\pm 1 \pm i,\pm 1\pm i) \in \C^2$,
  and where $B_e \in \Hom_\R(\R^4,\C^2)$, $v \mapsto (iv_0\sigma_0
  + v_1\sigma_1 + v_2\sigma_2 + v_3\sigma_3)^Te$.
\qed\end{proof}

\section{Gauges, usage}

The concept of a gauge in \defref{gauge}
can be applied at both
(i) the linear and formal perturbative nonlinear level
and (ii) the nonlinear level.
At the level (i)
we get a contraction
for a dgLa that, via the machinery of $L_\infty$-homotopy transfer,
is directly applicable at the formal perturbative nonlinear level.
At the level (ii)
we get local-in-time existence and uniqueness for the Einstein 
equations, a standalone alternative
to the traditional approach using the harmonic gauge
of Einstein and, rigorously, Y.~Choquet-Bruhat.
Here we limit ourselves to (i) because it
relates to the homological framework,
and because the same manipulations also yield (ii).

\defref{gauge} is purely algebraic,
whereas symmetric hyperbolicity
is usually defined using explicit
matrix notation as in \secref{shs}.
The following theorem and proof show how they are brought together,
via the anchor map.
Recall $\Hom = \Hom_\RR$.
\begin{theorem}[Linear symmetric hyperbolic system] \label{theorem:lshs}
  Suppose a gauge $(\E_G,B)$ is fixed, see \defref{gauge}.
  Suppose an element $\unk \in \E^1$ is fixed,
  and suppose it is globally hyperbolic in the sense of \defref{angh}.
  For every $k$ define
  \begin{align*}
      L^k: \E^k_G & \to \Hom(\E^k_G,\RR)\\
      \qquad u & \mapsto B^k(-,[\unk,u])
  \end{align*}
Then for every fixed $R \in \Hom(\E^k_G,\RR)$,
the equation $L^k(u)=R$
is a linear symmetric hyperbolic 
PDE for the unknown
$u \in \E^k_G$, when written out in a suitable coordinate system
$M \simeq \R^4$, and relative to a $\RR$-basis for $\E_G^k$.
The map $L^k$ is surjective,
and the kernel of $L^k$ is isomorphic to
restrictions of elements of $\E^k_G$ to $t=0$.
\end{theorem}
\begin{proof}
  We suppress the index $k$, and we note
  that the right hand side $R$ is irrelevant for symmetric
  hyperbolicity.
The map $L$ is a first order differential operator,
in the sense that for every $f \in \RR$
the map $J_f(u) = L(fu) - fL(u)$ is $\RR$-linear,
\[
J_f \in \Hom(\E_G,\Hom(\E_G,\RR))
\simeq \Hom(\E_G \otimes \E_G,\RR)
\]
In fact there is
an $a \in W \otimes \Der(\RR)$
with
$[\unk,fu] = a(f) u + f [\unk,u]$
for all $f \in \RR$ and $u \in \E$,
a piece of the anchor map,
so $J_f = B(-,a(f)-)$
with $a(f) \in W$.
\defref{gauge} implies $J_f \in \Hom(S^2 \E_G,\RR)$,
the symmetry condition for a symmetric hyperbolic
equation. For the positivity condition,
use the coordinate system $M \simeq \R^4$
that yields global hyperbolicity in \defref{angh},
with $t \in \RR$ the first coordinate.
Then $a(t) \in W_+$,
and therefore $J_t$ is positive definite by \defref{gauge}.
The surjectivity and kernel
follow from global solvability of linear symmetric hyperbolic equations.
\qed\end{proof}
\begin{theorem}[Contraction]
  With the assumptions of \theoremref{lshs},
  in particular global hyperbolicity,
  the following composition is surjective for every $k$:
  \[
      \oldw:\;\; \E_G^k \inj
      \E^k \xrightarrow{\;\;[\unk,-]\;\;}
      \E^{k+1}
      \surj
      \E^{k+1}/\E_G^{k+1}
    \]
    If in addition $\unk \in \MC(\E)$
    and $d = [\unk,-]$ the associated differential, then
    there is a contraction
    from $(\E,d)$ down to the subcomplex $(\ker \oldw,d|_{\ker \oldw})$.
    A homotopy giving the contraction is given by the composition
  \[
     \E^{k+1} \surj \E^{k+1}/\E_G^{k+1}
     \to \E_G^k \inj \E^k
   \]
   where the middle arrow is any $\R$-linear
   (not $\RR$-linear) right-inverse of $\oldw$.
\end{theorem}
\begin{proof}
  Every $r \in \E^{k+1}/\E_G^{k+1}$ 
  yields a well-defined $R = B^k(-,r) \in \Hom(\E_G^k,\RR)$,
  so surjectivity follows from \theoremref{lshs}
  and \defref{gauge}.
\qed\end{proof}

\section{The constrained structure of $\Z_2$-graded Clifford modules}\label{sec:cmod}
The Clifford algebra construction \cite{abs,sping} is a functor
from finite-dimensional real inner product spaces
to finite-dimensional unital
associative real algebras with a distinguished subspace.
Let $\Cl_{pq}$ be the real Clifford algebra
with $p$ respectively $q$ generators squaring to $+1$ respectively $-1$.
The generators $e_i$ are understood to satisfy $(e_i)^2 = \pm 1$
and $e_ie_j + e_je_i = 0$ if $i \neq j$.
The distinguished subspace is the span of the $p+q$ generators.
There is a canonical $\Z_2$-grading by declaring that the distinguished
subspace be odd. The Clifford algebra
has a canonical non-decreasing filtration.

In general $\Cl_{pq}$ is not isomorphic to $\Cl_{qp}$ as a real algebra,
but this is inconsequential if one studies $\Z_2$-graded modules;
all Clifford modules in this paper are.
All modules are understood to be finitely generated, unital left modules.
\begin{lemma}[Category of $\Z_2$-graded modules]
  The $\Z_2$-graded algebras $\Cl_{pq}$ and $\Cl_{qp}$ have the same categories of
  $\Z_2$-graded modules.
\end{lemma}
\begin{proof}
  To avoid misconceptions, $p \neq q$.
  Let $e_i$ be the generators of $\Cl_{pq}$
  and $f_i$ the generators of $\Cl_{qp}$.
  Order them such that $(e_i)^2 = 1$ if and only if $(f_i)^2 = -1$.
  Let $M$ be a $\Z_2$-graded module of $\Cl_{pq}$.
  Let $s \in \End(M)$ be equal to $1$ respectively
  $-1$ on the even respectively odd sector of $M$.
  Then $M$ becomes a $\Z_2$-graded module of $\Cl_{qp}$
  by representing $f_i$ as $e_i s$.
  To see this,
  observe that $s^2 = 1$ and, since the $e_i$ are represented as odd elements,
  $s e_i + e_is = 0$.
  We have only discussed the correspondence at the level of objects,
  but it is easily extended to morphisms.
  As a strict aside, by viewing $s$ as a new Clifford generator,
  this proof establishes an isomorphism
  $\Cl_{p+1,q} \simeq \Cl_{q+1,p}$.
\qed\end{proof}

The structure of $\Z_2$-graded Clifford modules
is highly constrained, much more so than for the exterior algebra.
We only consider $\Cl_{13}$.
Its even subalgebra is isomorphic
as just algebras to
$\Cl_{30}$. Let $S \in \Aut_\R(\Cl_{30})$ be the outer real algebra
automorphism
that acts like minus the identity on the three generators of $\Cl_{30}$.
\begin{lemma} \label{lemma:cl30}
  A $\Cl_{30}$-module is free iff it extends
  to a module of the real algebra given by the presentation
  $\langle \Cl_{30}, T \mid T^2 = \mathbbm{1},\,\text{$S(m)=TmT$ for all 
    $m \in \Cl_{30}$}\rangle$
    where the symbol $T$ is a new generator.
\end{lemma}
\begin{proof}
We have
$\Cl_{30} \simeq M_2(\C)$ as real algebras
by mapping the three generators
$x_1,x_2,x_3 \in \Cl_{30}$
to the three Pauli matrices $\sigma_1, \sigma_2, \sigma_3 \in M_2(\C)$.
The automorphism $S$
corresponds to conjugating elements of $M_2(\C)$
by the quaternionic matrix
$J = \smash{(\begin{smallmatrix} 0 & j \\ -j & 0 \end{smallmatrix})}$,
an involution.
The algebra presented in the lemma is $M_2(\HQ)
= M_2(\C) \oplus M_2(\C)J$.
Modules of $M_2(\HQ)$ are isomorphic to $(\HQ^2)^n$ for some $n$,
with $\HQ^2$ the quaternionic column vectors,
and $\HQ^2 \simeq M_2(\C)$ as $M_2(\C)$-modules, so it is free.
\qed\end{proof}
Let $e_0$, $e_1$, $e_2$, $e_3$ be the generators
of $\Cl_{13}$, in particular $(e_0)^2 = 1$.
  Let $P \in \Aut_\R(\Cl_{13})$ be 
  the algebra automorphism induced by
  $(e_0,e_1,e_2,e_3) \mapsto (e_0,-e_1,-e_2,-e_3)$.
\begin{theorem}[Characterization of free $\Z_2$-graded modules of $\Cl_{13}$] \label{theorem:pf}
  A $\Z_2$-graded $\Cl_{13}$-module is free
  iff it extends to a $\Z_2$-graded module of the real algebra
  presented by
  $\langle \Cl_{13}, T \mid \text{$T$ even},\, T^2 = \mathbbm{1},\,\text{$P(m)=TmT$ for all 
    $m \in \Cl_{13}$}\rangle$
    with $T$ a new symbol.
\end{theorem}
  Being free means the module is isomorphic
  to a power of $\Cl_{13}$ 
  as a $\Z_2$-graded $\Cl_{13}$-module;
  the isomorphism need not encompass
  the filtrations if the module is filtered.

Beware that $T$ has to be even.
Otherwise the existence of such an operator is trivial
because $P$ is inner, $P(m) = e_0me_0$,
yet not all $\Z_2$-graded modules are free, for instance $\Cl_{13}$ itself
is a direct sum of two proper submodules as a module over itself.

All $\Cl_{13}$-modules in this paper are $\Z_2$-graded and naturally
come with an operator $T$,
and all morphisms respect this,
so \theoremref{pf} is quite useful.
\begin{proof}
We only prove the `if' claim. Let $M$ be the module.
  An algebra isomorphism $\Cl_{30} \to \Cl_{13}^\even$ is defined by $x_i \mapsto e_0e_i$.
  View $N = M^\even$ as a $\Cl_{30}$-module, note that $T(N) \subset N$
  and use \lemmaref{cl30}, so $N$ is free.
  We have $M \simeq N \oplus N$ as $\Z_2$-graded $\Cl_{13}$-modules,
  with the opposite $\Z_2$-grading in the second direct summand,
  where the module structure of $N \oplus N$
  is such that $\Cl_{30}$ acts
  diagonally, and $e_0$ exchanges summands, and recall $(e_0)^2 = 1$.
  Conclude that $N \oplus N \simeq M$ is free.
\qed\end{proof}

Every Clifford module defines, and is defined by, a representation of 
a finite group called the Clifford group.
This allows one to 
bring finite group techniques to bear.
\begin{lemma}[The finite Clifford group]\label{lemma:fcg}
  For a choice of generators $\{ e_i \} \subset \Cl_{pq}$,
  the submonoid generated by $\{\pm 1, e_i\}$
  is a group $F$ of finite order $|F| = 2^{p+q+1}$.
  Each element is $\Z_2$-odd or $\Z_2$-even.
  Every $\Cl_{pq}$-module restricts to a real representation
  of $F$ that represents $-1$ as minus the identity,
  and this is a one-to-one correspondence.
  For every $i$
  there is a unique character $\chi_i: F \to \{\pm 1\}$
  defined by $f e_i = \chi_i(f) e_i f$.
\end{lemma}
\begin{proof}
  Omitted.
\qed\end{proof}
The group $F$ depends on the choice of generators,
but it allows us to define an object that does not,
for $\Cl_{1q}$.
\begin{theorem}[Invariant Clifford average in $\Cl_{1q}$] \label{theorem:ica}
  Define $\pi \in S^2 \Cl_{1q}$ by
  \[
      \pi = \frac{1}{|F|} \sum_{f \in F} \chi_0(f)\, f \otimes f
  \]
  Then $\pi$ is invariant in the sense that
  it is independent of the choice of generators used to define $F$. 
  In the $\Z_2$-graded algebra $S^2 \Cl_{1q}$ we have:
  \begin{itemize}
    \item $\pi^2 = \pi$ and $\pi$ is even.
    \item $\pi(x \otimes 1) = \pi (1 \otimes x)$ for all $x$  in
      the distinguished subspace, $ x \in \R^{1+q}\subset \Cl_{1q}$.
    \item $\pi(1 \otimes e_0)
      = \frac{1}{|F|} \sum_{f \in F} f \otimes e_0 f$
      with $e_0$ the first basis element used to define $F$.
  \end{itemize}
\end{theorem}
\begin{proof}
  The proof of invariance is omitted,
  but the idea is that say
  $\sum_i \chi_0(e_i) e_i \otimes e_i
  = e_0 \otimes e_0 - e_1 \otimes e_1
  - \ldots - e_q \otimes e_q$ is invariant.
  By construction $\pi(f \otimes f) = \chi_0(f) \pi$ for all $f \in F$
  which implies $\pi^2 = \pi$.
  Also
  $\pi(f \otimes f^2) = \chi_0(f) \pi(1\otimes f)$.
  Set $f = e_i$ and note that we happen to have
  $(e_i)^2 = \chi_0(e_i)$ and therefore
  $\pi(e_i \otimes 1) = \pi(1 \otimes e_i)$,
  hence $\pi(x \otimes 1) = \pi(1 \otimes x)$ by linearity. The rest is clear.
\qed\end{proof}

\begin{remark}
The Clifford algebra is filtered,
$\Gr \Cl_{13} \simeq \wedge \R^4$ as graded commutative algebras.
The associated graded $\Gr$ is a functor
from filtered $\Z_2$-graded $\Cl_{13}$-modules
to graded $\wedge \R^4$-modules.
One can ask which $\wedge \R^4$-modules and morphisms are in the image of the $\Gr$-functor,
and which $\wedge \R^4$-modules
are the associated gradeds
of $\Cl_{13}$-modules
that as unfiltered modules are free as in \theoremref{pf}.
A necessary condition
is that the real dimension has to be a multiple of $\dim_\R \Cl_{13} = 16$.
Though free $\wedge \R^4$-modules are in the image, some non-free modules are too.
\end{remark}
\section{Symmetric hyperbolic systems} \label{sec:shs}

The theorem of Picard-Lindel\"of gives local
existence and uniqueness for ODE.
There is a similar theorem for a class of
PDE called quasilinear symmetric hyperbolic systems.
We only discuss local control and hence use germs;
global control
requires a more problem specific analysis,
just as it does for ODE.

Let $x^\mu$ and $\p_\mu$ be the standard coordinates and partial derivatives
on $\R^n$.
Denote by
$\Hermitian_k \subset \C^{k \times k}$
the real subspace of Hermitian matrices.
\begin{theorem}[Local existence and uniqueness] \label{theorem:shs}
  Suppose
  $A^\mu \in C^{\infty}(\R^n \times \C^k, \Hermitian_k)$
  for $\mu = 0,\ldots,n-1$
  and
  $b \in C^{\infty}(\R^n \times \C^k, \C^k)$.
Suppose $A^0(0,0)$ is positive definite.
Then there exists a unique $u \in C^{\infty}_\textnormal{germs at $0$}(\R^n,\C^k)$
such that, as germs at $x=0$,
\[
  \left\{
    \begin{array}{l}
      \tsum_\mu A^\mu(x,u(x))\,(\p_\mu u)(x) = b(x,u(x))\\
      \rule{0pt}{13pt} 
      u|_{x^0=0} = 0
    \end{array}
    \right.
\]
\end{theorem}
\begin{proof}
  Omitted, see \cite{FrKO,taylor}.
  Briefly, one derives a-priori energy estimates by
  applying the divergence theorem
  to $\smash{\tsum_\mu \p_\mu(u^\ast A^\mu(x,u) u)}$
  and higher derivative expressions.
\qed\end{proof}

Beware that even if $A^0 = \mathbbm{1}$,
  the claim fails if the $A^{\mu}$ are not in
  $\Hermitian_k$, see Lewy's example.
We have assumed that $A^\mu$ and $b$ are smooth
and everywhere defined,
that $u$ satisfies trivial initial conditions at $x^0=0$, and so forth.
This simplified statement implies more general statements,
say by changing coordinates in $x$ and $u$.
ODE correspond to $n=1$.
An interesting example related to Maxwell's equations is
$\tsum_\mu A^\mu\p_\mu = \p_0 + i\curl$ with $n=4$, $k=3$.
Unlike parabolic equations,
symmetric hyperbolic systems enjoy
finite speed of propagation.

\section{Elements of Maurer-Cartan perturbation theory} \label{sec:mcpt}

See Gerstenhaber \cite{Gerstenhaber}.
We describe the unobstructed case, for
any gLa free over $\R[[s]]$.
Here $s$ is a symbol,
analogous statements hold for several symbols.

\begin{definition}[gLa free over ${\R[[s]]}$] \label{def:gfo}
  We say that $\pxg$ is a gLa free over $\R[[s]]$ if
  $\pxg$ is a gLa over $\R[[s]]$
  and
  there is a graded $\R$-vector space $\amcb$
  and an isomorphism of graded $\R[[s]]$-modules $\pxg \simeq \amcb[[s]]$.
  The induced $\pxg/s\pxg \simeq \amcb$ makes $\amcb$ a real gLa.
\end{definition}

The bracket on $\pxg$ is the $\R[[s]]$-bilinear extension
of a map $\amcb \times \amcb \to \amcb[[s]]$,
not necessarily
$\amcb \times \amcb \to \amcb$.
Informally, the bracket can itself insert powers of $s$.

An example is when
$\pxg$ is the Rees algebra of a filtration of a real gLa, that is,
$\gx$ is a real gLa with a non-decreasing gLa-filtration
$(F_p\gx)_{p \geq 0}$ 
with $F_p \gx = \gx$ for almost all $p$,
and $\pxg = \{ \sum_p x_p s^p \in \gx[[s]] \mid x_p \in F_p \gx \}$.
Then $\pxg/s\pxg$ is the associated graded gLa.

\begin{theorem}[The unobstructed case]
  Suppose $\pxg$ is a gLa free over $\R[[s]]$.
  Suppose $x_0 \in \MC(\gsg)$. Define
  the differential $d = [x_0,-] \in \End^1(\gsg)$ and set
  \[
      \MC_{x_0}(\pxg) = \{
      x \in \MC(\pxg) \mid x = x_0 \bmod s\pxg^1 \}
  \]
Write $H^k = H^k(d)$
for the $k$-th homology.
Suppose $H^2 = 0$ (`unobstructed'). Then:
\begin{itemize}
  \item There exists a map
    $\phi : H^1 \to \smash{\MC_{x_0}(\pxg)}$
of the form
$\phi(\yy) = \smash{x_0 + \tsum_{k \geq 1} s^k \phi_k(\yy^{\otimes k})}$
where $\smash{\phi_k \in \Hom_{\R}((H^1)^{\otimes k},\pxg^1)}$
and $\phi_1(\yy) \bmod s \pxg^1$ is a representative of $\yy \in H^1$.
\item
Every such $\phi$ extends, by $\R[[s]]$-multilinear
extension of $\phi_k$,
to 
$H^1[[s]] \to \MC_{x_0}(\pxg)$, and this map induces a bijection
onto `the formal moduli space at $x_0$':
\[
  H^1[[s]] \to
  \frac{\MC_{x_0}(\pxg)}{\exp(s \pxg^0)}
  \]
  \end{itemize}
\end{theorem}
\begin{proof}
  Freeness
  is used whenever we invoke the isomorphism
  $1/s^K : s^K \pxg \to \pxg$.
  By $H^2 = 0$ there is an
  $h : \amcb^2 \to \amcb^1$ with $dh|_{\amcb^2 \cap \ker d} = \mathbbm{1}$.
  Let $i: \amcb \inj \pxg$ and
  $p: \pxg \surj \amcb$ be the canonical maps,
$pi = \mathbbm{1}$.
Let $r: H^1 \to \amcb^1$ choose representatives.
For each $\xi \in H^1$
we construct $c_k \in \pxg^1$ such that
$e_K \in s^{K+1} \pxg^2$ for all $K$,
where by definition
$e_K = [\Xi_{\leq K},\Xi_{\leq K}]$
and
$\Xi_{\leq K} = \sum_{k \leq K} s^k c_k$.
Set $c_0 = ix_0$,
set $c_1 = -\tfrac{1}{2} ihp(e_0/s) + ir\xi$,
and thereafter set
$c_{K+1} = -\tfrac{1}{2}ihp(e_K/s^{K+1})$.
We show by induction on $K$:
\[
  A_K:\;\; e_K \in s^{K+1} \pxg^2
  \qquad
  B_K:\;\; dp(e_K/s^{K+1}) = 0
  \qquad
  C_K:\;\; dpc_{K+1} = -\tfrac{1}{2} p(e_K/s^{K+1})
\]
  Here, $A_0$ holds by $x_0 \in \MC(\amcb)$; and $A_K$ by
  $e_K = e_{K-1} + 2 s^K [c_0, c_K] \mod s^{K+1} \pxg^2$
  and $C_{K-1}$.
  Next, $B_K$ by
  $[c_0,e_K] = [c_0-\Xi_{\leq K},e_K]
    + [\Xi_{\leq K},e_K] \in s^{K+2}\pxg^3$
  where the first term is in $s^{K+2}\pxg^3$ by
  $A_K$, the second is zero by a Jacobi identity.
  Finally $C_K$ holds by $B_K$
  and the definition of $h$; for $C_0$ use
  $dr\xi = 0$.
  This map $\xi \mapsto \sum_{k \geq 0} s^k c_k$
  is a map $\phi$ of the desired kind;
  the $c_k$ are not homogeneous in $\xi$ but
  one can reorganize
  to extract homogeneous $\phi_k(\xi^{\otimes k})$.
  The $\R[[s]]$-multilinear extension of a given $\phi$ is clear.
\qed\end{proof}


\section{Spinor functor}\label{sec:spinors}

By the `spinor functor' we mean a functor from the groupoid of 2-dimensional
complex vector spaces to the groupoid of 4-dimensional real vector
spaces with a conformal inner product of signature ${-}{+}{+}{+}$;
the morphisms are the structure-preserving isomorphisms,
and a conformal inner product is one modulo $\R_+$.
The spinor functor associates to the 2-dimensional complex $V$
the 4-dimensional real subspace\footnote{%
The conjugate $\overline{V}$ is a vector space together with $\C$-antilinear
maps $V \to \overline{V}$ and $\overline{V} \to V$ that are
mutual inverses; it exists and is unique up to isomorphism.
Conjugation on $V \otimes \overline{V}$
is $x\otimes y \mapsto \overline{y} \otimes \overline{x}$.
}
\[
  W_V \subset V \otimes_\C \overline{V}
\]
with a representative $S^2 W_V\to \R$ of the conformal inner product
the restriction of
the canonical
$S^2(V \otimes \overline{V}) \to (\wedge^2 V) \otimes (\wedge^2 \overline{V})
 \simeq \C$, which has the right signature.

Applying this fiberwise yields a corresponding `spinor functor'
\[
  \ggrs \to \ggr
\]
where on the left we have the groupoid of
rank 2 complex vector bundles over a base manifold $\simeq \R^4$;
the morphisms are the isomorphism of vector bundles.
We denote by $V$ the $\CR = \C \otimes \RR$-module of sections of this vector bundle,
by $W_V$ the associated $\RR$-module of rank $4$
with conformal inner product.

\begin{lemma}[Module derivations of $V$]
  Let
  $\DerEnd_\RR(V) \subset \Der(\RR) \oplus \End_{\C}(V)$ be the module derivations of $V$
  as a $\RR$-module (not $\CR$) that are $\C$-linear on $V$.
  Then there are canonical $\RR$-Laoid morphisms:
  \begin{itemize}
    \item
  $\DerEnd_\RR(V) \to \DerEnd_\RR(\oV)$, 
  actually an isomorphism.
\item $\DerEnd_\RR(V) \to \CDerEnd(W_V)$,
  surjective with kernel
  the $\RR$-span of $0 \oplus i\mathbbm{1}$.
  \end{itemize}
\end{lemma}
\begin{proof}
  The first map is $\delta \mapsto \delta' = c \circ \delta \circ c$
  where $c$ is conjugation,
  and it is the identity on $\Der(\RR)$.
  The second is $\delta \mapsto 
  (x \otimes y \mapsto
  \delta x \otimes y + x \otimes \delta' y)$
  which is well-defined
  with, for once, the tensor products over $\CR$.
\qed\end{proof}
We give an equivalent definition of the spinor functor using a basis.
\begin{lemma}
    If $V = \CR v \oplus \CR w$
    then a conformally orthonormal frame for $W_V$ is
    \begin{align*}
      \theta_0 & = v\overline{v}+w\overline{w} &
      \theta_1 & = v\overline{w}+w\overline{v} & 
      \theta_2 & = iw\overline{v}-iv\overline{w} &
      \theta_3 & = v\overline{v}-w\overline{w}
    \end{align*}
Define
$\sigma_0,\sigma_1,\sigma_2,\sigma_3,\sigma_{23},\sigma_{31},\sigma_{12}
\in \End_{\CR}(V) \cap \DerEnd_\RR(V)$ by
    \begin{align*}
      \sigma_0 & = \tfrac{1}{2} (\begin{smallmatrix}
        1 & 0 \\ 0 & 1
      \end{smallmatrix})
      &
      \sigma_1 & = \tfrac{1}{2} (\begin{smallmatrix}
        0 & 1 \\ 1 & 0
      \end{smallmatrix})
      &
      \sigma_2 & = \tfrac{1}{2} (\begin{smallmatrix}
          0 & -i \\ i & 0
      \end{smallmatrix})
      &
      \sigma_3 & = \tfrac{1}{2} (\begin{smallmatrix}
        1 & 0 \\ 0 & -1
      \end{smallmatrix})\\
      &&
      \sigma_{23} & = \tfrac{1}{2} (\begin{smallmatrix}
        0 & i \\ i & 0
      \end{smallmatrix})
      &
      \sigma_{31} & = \tfrac{1}{2} (\begin{smallmatrix}
        0 & 1 \\ -1 & 0
      \end{smallmatrix})
      &
      \sigma_{12} & = \tfrac{1}{2} (\begin{smallmatrix}
        i & 0 \\ 0 & -i
      \end{smallmatrix})
    \end{align*}
    relative to the basis $v,w$.
    This is consistent, so
    under $\DerEnd_\RR(V) \to \CDerEnd(W_V)$
    these elements map to elements of the same name in \defref{uzt}.
    \end{lemma}
    \begin{proof}
      Omitted.
      \qed\end{proof}
Clearly $\L_V = (\wwv) \otimes \DerEnd_\RR(V)$
is a $\wwv$-gLaoid via base change $\RR \inj \wwv$.
Base change gives
$\wwv$-gLaoid representations
$\L_V \to \DerEnd_{\wwv}(M_k)$ where
\[
  M_k = (\wwv) \otimes (\wedge_{\CR}^k V)
  \simeq (\wedge_{\CR} (V \otimes_\CR \overline{V})) \otimes_{\CR} (\wedge_\CR^k V)
\]
We can use this to construct the ideal $\I \subset \L$.
\begin{lemma}
  Set $N = (\wwv)N^2$ with $N^2
  = \SPAN
  \{(v\overline{x} \wedge_\CR v \overline{y}) \otimes_\CR v \mid v,x,y \in V\}$.
  Then $N \subset M_1$ is an $\L_V$-invariant $\wwv$-submodule,
  hence $M_1/N$ a representation.
  Let $\I_V$ be the kernel of the
  gLaoid representation $\L_V \to \DerEnd_{\wwv}(M_0 \oplus M_1/N \oplus M_2)$.
  Then the image of $\I_V$ under the surjection $\L_V \to \L
  = (\wwv) \otimes \CDerEnd(W_V)$
  is $\I$.
\end{lemma}
\begin{proof}
  Omitted. 
\qed\end{proof}


\newcommand{\href}[2]{#2}

\end{document}